\documentclass[a4paper,11pt]{article}

\usepackage[utf8]{inputenc}
\usepackage[T1]{fontenc}
\usepackage[UKenglish]{isodate}

\usepackage{mathrsfs}
\usepackage{bbm}
\usepackage{palatino}

\usepackage[margin=3cm]{geometry}

\usepackage[usenames,dvipsnames]{xcolor}
\usepackage{amsfonts}
\usepackage{amsthm}  
\usepackage{amsmath}
\usepackage{amssymb}
\usepackage{thmtools}
\usepackage{ifthen}
\usepackage{mathabx}
\usepackage{stmaryrd}
\usepackage[colorlinks=true,urlcolor=Mahogany,linkcolor=Mahogany,citecolor=Mahogany,plainpages=false,pdfpagelabels]{hyperref}
\hypersetup{pdftitle={Template}}
\usepackage{tikz}
\usepackage{verbatim}
\usetikzlibrary{shapes.geometric,plotmarks,backgrounds,fit}
\usepackage[backend=biber,style=alphabetic,maxcitenames=4,maxalphanames=100,maxbibnames=100,isbn=false]{biblatex}
\usepackage{todonotes}
\usepackage{authblk}

\newtheorem{theorem}{Theorem}
\newtheorem{lemma}[theorem]{Lemma}

\newtheorem{definition}[theorem]{Definition}

\newcommand{\sket}[1]{{\ensuremath{\lvert#1\rangle}}}
\newcommand{\lket}[1]{{\ensuremath{\left\lvert#1\right\rangle}}}
\newcommand{\ket}[1]{\mathchoice{\lket{#1}}{\sket{#1}}{\sket{#1}}{\sket{#1}}}

\newcommand{\sbra}[1]{{\ensuremath{\langle#1\rvert}}}
\newcommand{\lbra}[1]{{\ensuremath{\left\langle#1\right\rvert}}}
\newcommand{\bra}[1]{\mathchoice{\lbra{#1}}{\sbra{#1}}{\sbra{#1}}{\sbra{#1}}}

\newcommand{\sketbra}[2]{{\ensuremath{\lvert #1\rangle\langle #2\rvert}}}
\newcommand{\lketbra}[2]{{\ensuremath{\left\lvert #1 \middle\rangle\middle\langle #2\right\rvert}}}
\newcommand{\ketbra}[2]{\mathchoice{\lketbra{#1}{#2}}{\sketbra{#1}{#2}}{\sketbra{#1}{#2}}{\sketbra{#1}{#2}}}

\newcommand{\proj}[1]{\ketbra{#1}{#1}}

\DeclareFieldFormat{eprint:iacr}{%
  IACR eprint\addcolon\space
  \ifhyperref
    {\href{http://eprint.iacr.org/#1}{\texttt{#1}}}
    {\texttt{#1}}}

\DeclareFieldFormat{eprint:hal}{%
  HAL Id\addcolon\space
  \ifhyperref
      {\href{https://hal.archives-ouvertes.fr/#1}{\texttt{#1}}}
    {\texttt{#1}}}

\newcommand{\ident}{\mathbbm{1}}
\DeclareMathOperator{\tr}{\mathrm{Tr}}

\DeclareMathOperator{\Herm}{Herm}

\newcommand{\hmin}{H_{\min}}

\newcommand{\mbC}{\mathbb{C}}

\newcommand{\mbE}{\mathbb{E}}

\newcommand{\mbR}{\mathbb{R}}

\newcommand{\sfA}{\mathsf{A}}
\newcommand{\sfB}{\mathsf{B}}
\newcommand{\sfC}{\mathsf{C}}

\newcommand{\sfE}{\mathsf{E}}

\renewcommand{\otimes}{\varotimes}

\newcommand{\linop}{\mathrm{L}}
\newcommand{\Pos}{\mathrm{Pos}}
\newcommand{\Normal}{\mathrm{N}}
\newcommand{\D}{\mathrm{D}}
\DeclareMathOperator{\Supp}{supp}
\newcommand{\wt}{\mathrm{wt}}

\tikzstyle{porte} = [fill=blue!20, draw]
\tikzstyle{portepale} = [fill=blue!10, draw=black!50, text=black!50, minimum width=1cm]
\tikzstyle{unitaire} = [fill=red!20, draw]
\tikzstyle{etatcarre} = [fill=green!20, draw]
\tikzstyle{etiquette} = [font=\scriptsize]
\tikzstyle{ket} = [shape=semicircle, shape border rotate=270, fill=green!20, draw]
\tikzstyle{bra} = [shape=semicircle, shape border rotate=90, fill=green!20, draw]
\tikzstyle{etat} = [fill=yellow!20, draw]

\begin{document}
\cleanlookdateon
\title{Privacy amplification and decoupling without smoothing}
\author{Frédéric Dupuis\\ \textit{Département d'informatique et de recherche opérationnelle}\\ \textit{Université de Montréal}\\ \textit{Montréal, Québec}}


\maketitle

\begin{abstract}
    We prove an achievability result for privacy amplification and decoupling in terms of the sandwiched Rényi entropy of order $\alpha \in (1,2]$; this extends previous results which worked for $\alpha=2$. The fact that this proof works for $\alpha$ close to 1 means that we can bypass the smooth min-entropy in the many applications where the bound comes from the fully quantum AEP~\cite{tcr08} or entropy accumulation~\cite{dfr16}, and carry out the whole proof using the Rényi entropy, thereby easily obtaining an error exponent for the final task. This effectively replaces smoothing, which is a difficult high-dimensional optimization problem, by an optimization problem over a single real parameter $\alpha$.
\end{abstract}

\section{Introduction}\label{sec:intro}
Consider the following scenario: we have a bipartite quantum state $\rho_{AE}$, with $A$ held by Alice and $E$ by an adversary Eve. Alice would like to apply an operation to $A$ to obtain a $C$ system that is uniform and independent from $E$. Here an ``operation'' means choosing a channel $\mathcal{R}^h$ at random from a family $\{ \mathcal{R}^h_{A \rightarrow C} : h \in \mathcal{H} \}$, announcing $h$ publicly and applying it to $A$. We are interested both in the case where $A$, $\mathcal{R}^h$ and $C$ are classical and in the case where they are quantum; the classical case corresponds to privacy amplification, and the quantum case to decoupling. In both cases, we would like to maximize the size of $C$ (in number of bits or qubits) that we get as a function of the input state $\rho_{AE}$. We have achievability results in the classical~\cite{KMR03,renner-phd} and quantum~\cite{state-merging,FQSW,dbwr10,fred-these} cases of the form
\begin{align}
    \mbE_h \left\| \mathcal{R}^h(\rho_{AE}) - \frac{\ident_C}{|C|} \otimes \rho_E \right\|_1 \leqslant 2^{\frac{1}{2}(\log |C| - H_2(A|E)_\rho)},
\end{align}
where $H_\alpha(A|E)_\rho$ denotes the \emph{sandwiched Rényi conditional entropy} of order $\alpha \in [\frac{1}{2}, 1) \cup (1, \infty)$ of $A$ given $E$ (see below for the definition). However, while the Rényi 2-entropy is what naturally comes out of these proofs, it has mostly been used as an auxiliary quantity, due to one major flaw: it is very sensitive to small variations in the state $\rho_{AE}$ over which it is computed, and as a result we can rarely get good bounds on it in most settings of interest. To remedy this, we instead use the min-entropy $\hmin(A|E)_\rho$ (which is a lower bound on the 2-entropy and has additional desirable properties---again, see below for the definition), and \emph{smooth} it, which means that we compute it on the best state in an $\varepsilon$-ball\footnote{The distance measure $D(\cdot,\cdot)$ used here is called the purified distance; its precise definition will not be needed here. Its definition and various properties can be found in, e.g.~\cite{livre-t15}} around $\rho$:
\[ \hmin^\varepsilon(A|B)_\rho := \sup_{\theta_{AB} : D(\rho,\theta) \leqslant \varepsilon} \hmin(A|B)_\theta. \]
This yields an achievability result of the form:
\begin{align}\label{eqn:bound-smooth-hmin}
    \mbE_h \left\| \mathcal{R}^h(\rho_{AE}) - \frac{\ident_C}{|C|} \otimes \rho_E \right\|_1 \leqslant 2\varepsilon + 2^{\frac{1}{2}(\log |C| - \hmin^\varepsilon(A|E)_\rho)}.
\end{align}
How can we get good lower bounds on $\hmin^\varepsilon$ in applications of interest? One simple case involves states where $A$ and $E$ both consist of $n$ independent, identically distributed systems $A_1^n$ and $E_1^n$. This is a major ``use case'', since many problems of interest (such as quantum key distribution protocols) can be reduced to it via the de Finetti theorem~\cite{renner-phd,r07}. In~\cite{tcr08}, the \emph{fully quantum asymptotic equipartition property} is proven, which states a lower bound in terms of the von Neumann entropy of a single system:
\[ \hmin^\varepsilon(A_1^n|E_1^n)_\rho \geqslant n H(A_1|E_1)_\rho - O(\sqrt{n}). \]
The proof works by first lower-bounding $\hmin^\varepsilon(A_1^n|E_1^n)$ by $H_\alpha(A_1^n|E_1^n) = nH_\alpha(A_1|E_1)$, and then lower-bounding $H_\alpha$ by $H$. Choosing $\alpha \approx 1 + \frac{1}{\sqrt{n}}$ then yields the bound we want. The same blueprint also works in more complicated settings: in~\cite{dfr16}, the entropy accumulation theorem (EAT), which can be used to bound the smooth min-entropy of states that are produced by step-by-step processes such as those arising in device-independent QKD protocols, also involves bounding the Rényi entropy for $\alpha$ just above 1.


In this paper, we show how to bypass the smooth min-entropy entirely in this type of proof and use a bound on the Rényi entropy of order $\alpha \in (1, 2]$ directly. In other words, we show that:
\begin{align}
    \mbE_h \left\| \mathcal{R}^h(\rho_{AE}) - \frac{\ident_C}{|C|} \otimes \rho_E \right\|_1 \leqslant 2^{\frac{2}{\alpha}-1} \cdot 2^{\frac{\alpha-1}{\alpha}(\log |C| - H_\alpha(A|E)_\rho)}.\label{eqn:main-intro}
\end{align}
This simplifies the overall analysis considerably: one gets much cleaner expressions, we can get error exponents quite easily on the final quantity of interest that are better than those that can be obtained via smoothing\footnote{See~\cite[Section III.D]{h10} for a comparison between error exponents obtained via the Rényi approach and the smooth min-entropy approach in the iid setting.}, and we replace smoothing (which is a difficult high-dimensional optimization problem) by an optimization over a single real parameter $\alpha$. This can then be applied to a wide variety of problems, ranging from privacy amplification in QKD, to decoupling~\cite{dbwr10,fred-these} in the quantum case, which in turn yields results in quantum channel coding. Furthermore, this approach raises an interesting conceptual question: can we do all of one-shot information theory using Rényi-type quantities instead of min/max quantities and smoothing?

The proof of the result is based on norm interpolation methods reminiscent of the Riesz-Thorin theorem. Specifically, we interpolate between the $\alpha=2$ case, which we know how to deal with using previously known bounds, and the $\alpha=1$ case, where a trivial bound is sufficient for our purposes.

The rest of this paper is structured as follows. In the next subsection, we give an overview of the state of the art and related results. Then, in Section~\ref{sec:prelims}, we present the notation we use as well as the necessary background on entropic quantities and norm interpolation, and give a definition of the class of randomizing families of channels to which our theorem applies. In Section~\ref{sec:main}, we prove the main result (i.e. Equation~\eqref{eqn:main-intro}), and Section~\ref{sec:entropy-accumulation} contains a sample application which shows what happens when we combine this result with the entropy accumulation theorem. Finally, Section~\ref{sec:discussion} concludes with a discussion of the result and of open questions.

\subsection{State of the art and related work}\label{sec:state-of-art}
In the purely classical case (i.e.~when $E$ is a classical random variable), an expression very similar to \eqref{eqn:main-intro} was derived by Hayashi~\cite[Equation (67)]{h10}, and tightness results are provided in~\cite{h10,hw16}. The results of~\cite{h10} were then extended to \emph{almost} 2-universal hash functions in~\cite{h16}. These results can then be used to obtain error exponents in the iid case.

In the CQ case (i.e.~classical $A$ but quantum $E$), the situation is more complicated. In~\cite{h12}, Hayashi gives two versions of \eqref{eqn:main-intro} that include additional terms involving either the number of distinct eigenvalues of an operator, or the ratio between the minimal and maximal eigenvalues of the same operator. In the iid case, this additional term is not problematic as it grows only polynomially in $n$ and is thus negligible as $n \rightarrow \infty$. In the general one-shot case, however, this term is much harder to deal with. In particular, such a result cannot be directly applied to a setting like DIQKD due to this issue. Another drawback of~\cite{h12} is that the bound obtained is in terms of the Petz-Rényi entropy rather than the sandwiched Rényi entropy, which is a worse bound in general.

In the fully quantum case, with both $A$ and $E$ quantum, Sharma~\cite{s15} obtains a result very similar to those of Hayashi in the CQ case, with an additional term involving the number of distinct eigenvalues of an operator. Sharma then goes on to apply this decoupling theorem to a panoply of quantum Shannon theory problems, obtaining error exponents for the iid case for all of them. However, the same caveat applies as in the CQ case: while the additional term poses no problem in the iid case, it is difficult to deal with in general.

\section{Preliminaries}\label{sec:prelims}

\subsection{Notation}\label{sec:notation}
In the table below, we summarize the notation used throughout the paper. 
\begin{center}
    \begin{tabular}{|c|l|}
        \hline
        \emph{Symbol} & \multicolumn{1}{c|}{\emph{Definition}}\\
        \hline
        $A, B, C, \dots$ & Quantum systems\\
        $\sfA,\sfB,\dots$ & Hilbert spaces corresponding to systems $A, B,\dots$\\
        $|A|$ & Dimension of $\sfA$\\
        $\linop(\sfA, \sfB)$ & Set of linear operators from $\sfA$ to $\sfB$\\
        $\linop(\sfA)$ & $\linop(\sfA, \sfA)$\\
        $X_{AB}$ & Operator in $\linop(\sfA \otimes \sfB)$\\
        $X_{A \rightarrow B}$ & Operator in $\linop(\sfA, \sfB)$\\
        $\Normal(\sfA)$ & Set of normal operators on $\sfA$\\
        $\Herm(\sfA)$ & Set of Hermitian operators on $\sfA$\\
        $\Pos(\sfA)$ & Set of positive semidefinite operators on $\sfA$\\
        $\D(\sfA)$ & Set of positive semidefinite operators on $\sfA$ with unit trace\\
        $\ident_A$ & Identity operator on $\sfA$\\
        $X^{\dagger}$ & Adjoint of $X$.\\
        $X_{A} \geqslant Y_A$ & $X-Y \in \Pos(\sfA)$.\\
        $\| X \|_p$ & Schatten $p$-norm of $X$: $\tr[(X^{\dagger} X)^{\frac{p}{2}}]^{\frac{1}{p}}$.\\
        $\Supp X$ & Support of $X$\\
        $H(A|B)_\rho$ & Conditional von Neumann entropy: $-\tr[\rho_{AB} \log \rho_{AB}] + \tr[\rho_B \log \rho_B]$\\
        \hline
    \end{tabular}
\end{center}
Note that all Hilbert spaces are finite-dimensional and endowed with a standard computational basis that we denote by $\{ \ket{a} : 1 \leqslant a \leqslant |A| \}$ for system $A$. Furthermore, by ``CQ operator'', we mean a bipartite operator $X_{AB}$ of the form:
\[ X_{AB} = \sum_{a} \proj{a} \otimes X_B(a) \]
where $X_B(a) \in \linop(\sfB)$ for all $a$.

\subsection{Entropic quantities}
Our main theorem involves the so-called \emph{sandwiched conditional quantum Rényi entropy} of a bipartite quantum state. This information measure was introduced relatively recently~\cite{mdsft13,wwy13} as an alternative way to generalize the classical Rényi entropy to the quantum case. It is defined as follows:

\begin{definition}[Sandwiched Rényi entropy]\label{def:sandwiched-renyi-entropy}
    Let $\rho_{AB} \in \D(\sfA \otimes \sfB)$ and $\sigma_B \in \D(\sfB)$, and let $\alpha \in [\frac{1}{2}, 1) \cup (1, \infty)$ be a real parameter. Then, the sandwiched Rényi entropy of order $\alpha$ of $A$ given $\sigma_B$ is defined as
    \begin{align*}
        H_\alpha(A|B)_{\rho|\sigma} := \begin{cases}%
            \frac{1}{1-\alpha} \log \tr\left[ \left( \sigma_B^{\frac{1-\alpha}{2\alpha}} \rho_{AB} \sigma_B^{\frac{1-\alpha}{2\alpha}} \right)^\alpha \right] & \text{if $\alpha < 1$ and $\tr[\rho\sigma] \neq 0$, or $\rho \in \sfA \otimes \Supp(\sigma)$}\\
            -\infty & \text{otherwise},
        \end{cases}
    \end{align*}
    and the sandwiched Rényi entropy of $A$ given $B$ is defined as: 
    \begin{align*}
        H_\alpha(A|B)_\rho := \max_{\omega_B \in \D(\sfB)} H_\alpha(A|B)_{\rho|\omega}.
    \end{align*}
\end{definition}
This quantity is monotone nonincreasing in $\alpha$ and satisfies the data processing inequality~\cite{fl13,b13}, and if we take the limit as $\alpha  \rightarrow  1$, we get the von Neumann entropy $H(A|B)_\rho$.  Furthermore, the limit of the Rényi entropy when $\alpha  \rightarrow  \infty$ makes sense and is called the min-entropy. It can be shown to be equal to
\begin{align*}
    \hmin(A|B)_\rho = H_\infty(A|B)_{\rho} = \sup \{ \lambda : \exists \sigma_B \in \D(\sfB), \rho_{AB} \leqslant 2^{-\lambda} \ident_A \otimes \sigma_B \}.
\end{align*}

Note also that the case $\alpha=2$ corresponds to the collision entropy, which is the quantity that shows up in the original proof of security of privacy amplification~\cite[][Chapter 5]{KMR03,renner-phd}, as well as in proofs of the ``decoupling theorem'' in the fully quantum case~\cite{FQSW,fred-these,dbwr10}. This is the main reason why the sandwiched Rényi entropy is relevant for us here: it offers a natural way to interpolate between the $\alpha=2$ case, which we know how to deal with, and the $\alpha=1$ case where a trivial bound suffices for our purposes.

While the Rényi 2-entropy is what naturally comes out of the above proofs, it has mostly been used as an auxiliary quantity, due to its sensitivity to small variations in the state $\rho_{AB}$ over which it is computed, and as a result we can rarely get good bounds on it in most settings of interest. To remedy this, we instead use the min-entropy (which is a lower bound on the 2-entropy by monotonicity in $\alpha$ and has additional desirable properties), and \emph{smooth} it, which means that we compute it on the best state in an $\varepsilon$-ball around $\rho$:
\[ \hmin^\varepsilon(A|B)_\rho := \max_{\theta_{AB} : D(\rho,\theta) \leqslant \varepsilon} \hmin(A|B)_\theta. \]

\subsection{Schatten $p$-norms and norm interpolation}\label{sec:norm-interpolation}
The main technical tool that we will need is a technique from complex analysis that allows us to interpolate between two norms. The relevant norms here are the Schatten $p$-norms:

\begin{definition}[Schatten $p$-norms]\label{def:schatten-p-norm}
    Let $N \in \linop(\sfA,\sfB)$ be a linear operator, and let $p \geqslant 1$. Then, the Schatten $p$-norm of $N$ is defined as
    \[ \| N \|_p := \tr\left[ (N^{\dagger}N)^{\frac{p}{2}} \right]^{\frac{1}{p}}. \]
\end{definition}

These norms satisfy the Hölder inequality: for $M \in \linop(\sfB,\sfC)$ and $N \in \linop(\sfA, \sfB)$ and $p>1$,
\[ \| MN \|_1 \leqslant \| M \|_p \| N \|_{\frac{p}{p-1}}. \]
It can also be shown that the following version holds for three operators $M, N, R$:
\begin{equation}\label{eqn:hoelder}
    \| MNR \|_1 \leqslant \| M \|_{p_1} \| N \|_{p_2} \| R \|_{p_3},
\end{equation}
where $1 = \frac{1}{p_1} + \frac{1}{p_2} + \frac{1}{p_{3}}$.
Furthermore, these norms can be expressed as the following optimization problem:
\begin{align*}
    \| M \|_p &= \sup_{Y \in \linop(\sfC, \sfB) : \| Y \|_{\frac{p}{p-1}} = 1} |\tr[YM]|\\
              &= \sup_{Y \in \linop(\sfC, \sfB) : \| Y \|_{\frac{p}{p-1}} = 1} |\tr[MY]|,
\end{align*}
and if $Q \in \Herm(\sfA)$, then,
\begin{align*}
    \| Q \|_p &= \sup_{Y \in \Herm(\sfA) : \| Y \|_{\frac{p}{p-1}} \leqslant 1} \tr[YQ].
\end{align*}
These norms are also unitarily invariant: for any unitaries $U_B$ and $V_C$,
\[ \| M \|_p = \| M U \|_p = \| V M \|_p. \]

The main reason for which these norms are relevant for us is that the sandwiched Rényi entropy can be expressed as:
\[ 2^{\frac{\alpha-1}{\alpha}H_\alpha(A|B)_{\rho|\sigma}} = \left\| \sigma_B^{\frac{1-\alpha}{2\alpha}} \rho_{AB} \sigma_B^{\frac{1-\alpha}{2\alpha}} \right\|_\alpha. \]

Our proof will be based on a Riesz-Thorin-like method to interpolate between the $2$-norm and the $1$-norm, in order to handle Rényi entropies with $\alpha \in (1,2]$. These methods have already been applied to quantum information theory before~\cite{pwpr06,db13}, and more specifically for proving properties of the sandwiched Rényi divergence~\cite{b13,d14,mt20}. Unfortunately, standard interpolation theorems such as Riesz-Thorin do not seem to apply directly to our case, and we will need to directly use the Hadamard three-line theorem, which is the main ingredient used to prove these standard results:

\begin{theorem}[Hadamard three-line theorem]\label{thm:hadamard-3lines}
    Let $S := \{ z \in \mbC : 0 \leqslant \Re(z) \leqslant 1 \}$, and let $f : S  \rightarrow  \mbC$ be a function that is holomorphic on the interior of $S$ and continuous on the border. Let $1 \leqslant p_0 \leqslant p_1$ and $0 < \theta < 1$, and define $p_0 \leqslant p_\theta \leqslant p_1$ via:
    \[ \frac{1}{p_\theta} = \frac{1-\theta}{p_{0}} + \frac{\theta}{p_1}.\]
Furthermore, for $k \in (0,1)$, let $M_k := \sup_{t \in \mbR} |f(k+it)|$. Then, for any $0 \leqslant \theta \leqslant 1$,
\[ |f(\theta)| \leqslant M_0^{1-\theta} M_1^{\theta}. \]
\end{theorem}
A proof can be found in~\cite[][page 33]{rs75}.

\subsection{Randomizing channels and hash functions}\label{sec:hashing}
Our basic setting is the following: we have a bipartite quantum state $\rho_{AE}$ where $E$ is a quantum system held by an adversary, and we want to apply a randomizing procedure to $A$ to get a $C$ such that the resulting state is close to uniform on $C$ and uncorrelated with $E$. To do this, we have a family of CPTP maps $\{ \mathcal{R}^h_{A \rightarrow C} : h \in \mathcal{H} \}$ together with a probability distribution $p$ on $\mathcal{H}$; we choose $h$ randomly according to $p$, publicly announce it (so that an adversary has access to it), and then apply the function to the system we want to randomize. A particular case of interest is when $A$ and $C$ are classical systems, in which case $\{ \mathcal{R}^h : h \in \mathcal{H} \}$ is a family of hash functions. In this case, each $h$ is a function from the set $\mathcal{A}$ to $\mathcal{C}$, and we define $\mathcal{R}^h_{A  \rightarrow  C}$ as:
\[ \mathcal{R}^h_{A  \rightarrow  C}(\theta_A) = \sum_{a \in \mathcal{A}} \bra{a} \theta_A \ket{a} \cdot \proj{h(a)}_C, \]
where the Hilbert spaces $\sfA$ and $\sfC$ have computational bases that are indexed by the sets $\mathcal{A}$ and $\mathcal{C}$ respectively.

We represent the setup in the following manner. Let $\mathcal{U}_{A \rightarrow C}$ be the perfectly randomizing channel
\[ \mathcal{U}(\theta_A) = \frac{\ident_C}{|\mathcal{C}|} \cdot \tr[\theta]. \]
Our goal is then to ensure that, for any state $\rho_{AE} \in \D(\sfA \otimes \sfE)$,
\[ \mbE_{h \sim p} \left\| (\mathcal{R}^h - \mathcal{U})(\rho_{AE}) \right\|_1 \leqslant \varepsilon, \]
for some suitably small $\varepsilon$. (Of course, in the CQ case, this condition only needs to be fulfilled for a CQ state $\rho_{AE}$.)

Our proof will be able to deal with families of channels that satisfy the following definition:

\begin{definition}[Randomizing family of channels]\label{def:rand-family}
    A family of channels $\{ \mathcal{R}^h_{A \rightarrow C} : h \in \mathcal{H} \}$ with distribution $p$ is \emph{$\lambda$-randomizing} if, for any $\rho_{AE} \in \linop(\sfA \otimes \sfE)$,
    \[ \mbE_{h \sim p} \| (\mathcal{R}^h - \mathcal{U})(\rho_{AE}) \|_2 \leqslant \lambda \| \rho_{AE} \|_2. \]
    In the CQ-case, a randomizing family of hash functions is $\lambda$-randomizing if the above condition holds for any CQ $\rho_{AE}$.
\end{definition}

Several families of channels that satisfy this definition exist, both fully quantum and CQ. For example, on the fully quantum side, quantum expanders~\cite{bst10,h07} are randomizing families of channels with an additional regularity condition. In addition, choosing a random unitary operator from a unitary 2-design and tracing out a subsystem is also a $1$-randomizing map~\cite{FQSW} (see also~\cite{fred-these,dbwr10}). On the CQ-side, the best known example consists of 2-universal families of hash functions, which were studied in the context of privacy amplification:
\begin{definition}[2-universal family of hash functions]\label{def:2-univ}
    A family of hash functions $\mathcal{H} \subseteq \mathcal{A}  \rightarrow  \mathcal{C}$ with distribution $p$ is called \emph{2-universal} if, for any $a, a' \in \mathcal{A}$ with $a \neq a'$,
    \[ \Pr_{h \sim p}\left[ h(a) = h(a') \right] = \frac{1}{|\mathcal{C}|}. \]
\end{definition}

We can show that any such family is 1-randomizing. This was proven in~\cite{KMR03,renner-phd} as part of the security proof of privacy amplification against quantum adversaries; we reproduce it here for convenience:
\begin{lemma}\label{lem:2-univ-is-2-norm-rand}
Any 2-universal family of hash functions is 1-randomizing.
\end{lemma}
\begin{proof}
    Let $\rho_{AE} \in \linop(\sfA \otimes \sfE)$ be a CQ operator, and write:
    \[ \rho_{AE} = \sum_a \proj{a}_A \otimes \rho_E(a). \]
    Then, we have:
    \begin{multline}
\mbE_h \| (\mathcal{R}^h - \mathcal{U})(\rho_{AE}) \|_2\\
        \begin{aligned}
         &= \mbE_h \tr\left[ (\mathcal{R}^h - \mathcal{U})(\rho) (\mathcal{R}^h - \mathcal{U})(\rho)^{\dagger} \right]^{\frac{1}{2}}\\
         &\leqslant \left( \mbE_h \tr\left[ (\mathcal{R}^h - \mathcal{U})(\rho) (\mathcal{R}^h - \mathcal{U})(\rho)^{\dagger} \right] \right)^{\frac{1}{2}}\\
         &= \left( \mbE_h \tr\left[ \mathcal{R}^h(\rho) \mathcal{R}^h(\rho)^{\dagger} - \left(\frac{\ident_C}{|C|} \otimes \rho_E\right) \mathcal{R}^h(\rho)^{\dagger} - \mathcal{R}^h(\rho)\left(\frac{\ident_C}{|C|} \otimes \rho_E\right)^{\dagger} + \frac{\ident_C}{|C|^2} \otimes \rho_E \rho_E^{\dagger}\right] \right)^{\frac{1}{2}}\\
         &= \left( \mbE_h \tr\left[ \mathcal{R}^h(\rho) \mathcal{R}^h(\rho)^{\dagger} \right] - \frac{1}{|C|} \tr[\rho_E \rho_E^{\dagger}] \right)^{\frac{1}{2}}, \label{eqn:2norm-bound}
    \end{aligned}
    \end{multline}
    Concentrating on the first term, we have:
    \begin{align*}
        \mbE_h  \tr\left[ \mathcal{R}^h(\rho) \mathcal{R}^h(\rho)^{\dagger} \right] &= \mbE_h \sum_{c \in \mathcal{C}} \sum_{a,a' \in \mathcal{A}, c=h(a)=h(a')} \tr[\rho_E(a)\rho_E(a')^{\dagger}]\\
                                                                              &= \mbE_h \sum_{a,a' \in \mathcal{A}, h(a)=h(a')} \tr[\rho_E(a)\rho_E(a')^{\dagger}]\\
                                                                              &= \sum_{a,a' \in \mathcal{A}} \Pr_h[h(a)=h(a')] \tr[\rho_E(a)\rho_E(a')^{\dagger}]\\
                                                                              &= \sum_a \tr[\rho_E(a) \rho_E(a)^{\dagger}] + \sum_{a \neq a'} \frac{1}{|C|} \tr[\rho_E(a)\rho_E(a')^{\dagger}]\\
                                                                              &\leqslant \|\rho_{AE} \|_2^2 + \frac{1}{|C|} \sum_{a,a'} \tr[\rho_E(a)\rho_E(a')^{\dagger}]\\
                                                                              &= \|\rho_{AE} \|_2^2 + \frac{1}{|C|} \tr[\rho_E\rho_E^{\dagger}].
    \end{align*}
    Substituting this into \eqref{eqn:2norm-bound} yields the lemma.
\end{proof}

\section{Main result}\label{sec:main}

As stated in the introduction, our main technical tool will be norm interpolation: we need a way to interpolate between the cases $\alpha=1$ and $\alpha=2$. Unfortunately, standard interpolation results cannot be immediately applied due to the expectation over $h$, and we need the following custom result that we prove directly from the three-line theorem:
\begin{lemma}\label{lem:3lines}
    Let $h$ be a random variable taking values in a set $\mathcal{H}$, and for each $h \in \mathcal{H}$, let $\mathcal{N}^h : \linop(\sfA)  \rightarrow  \linop(\sfB)$ be any superoperator, and let $\alpha \in (1,2]$. Then, for any $\rho_{AE} \in \D(\sfA \otimes \sfE)$ and any $\sigma_E \in \D(\sfE)$ such that $\left\| \sigma_E^{\frac{1-\alpha}{2\alpha}} \rho_{AE} \sigma_E^{\frac{1-\alpha}{2\alpha}} \right\|_\alpha = 1$, we have:
    \begin{multline*}
        \mbE_h \left\| \mathcal{N}^h\left(\sigma_E^{\frac{1-\alpha}{2\alpha}} \rho_{AE} \sigma_{E}^{\frac{1-\alpha}{2\alpha}} \right) \right\|_\alpha\\
        \leqslant \left( \max_{\eta \in \Normal(\sfA \otimes \sfE) : \| \eta_{AE} \|_1 = 1} \mbE_h \| \mathcal{N}^h(\eta_{AE}) \|_1 \right)^{\frac{2}{\alpha}-1} \left( \max_{\eta \in \Normal(\sfA \otimes \sfE) : \|\eta\|_2 = 1} \mbE_h \| \mathcal{N}^h(\eta_{AE}) \|_2 \right)^{2\left( \frac{\alpha-1}{\alpha} \right)}.
    \end{multline*}
    Furthermore, if $\rho_{AE}$ is a CQ state, then the maximizations in the two terms on the right-hand side can be restricted to CQ operators.
\end{lemma}
\begin{proof}
    For any $h \in \mathcal{H}$, let $Y_h \in \Herm(\sfB \otimes \sfE)$ such that $\| Y_h \|_{\frac{\alpha}{\alpha-1}} = 1$ and 
    \[ \tr\left[Y_h \mathcal{N}^h\left(\sigma_E^{\frac{1-\alpha}{2\alpha}} \rho_{AE} \sigma_{E}^{\frac{1-\alpha}{2\alpha}} \right)\right] = \left\|\mathcal{N}^h\left(\sigma_E^{\frac{1-\alpha}{2\alpha}} \rho_{AE} \sigma_{E}^{\frac{1-\alpha}{2\alpha}} \right) \right\|_\alpha. \]
        Now, define $f : S  \rightarrow  \mbC$ as:
    \begin{align*}
        f(z) := \mbE_h \tr\left[ Y_h^{z\left( \frac{\alpha}{2(\alpha-1)} \right)} \mathcal{N}^h\left(  \left( \sigma_E^{\frac{1-\alpha}{2\alpha}} \rho_{AE} \sigma_E^{\frac{1-\alpha}{2\alpha}} \right)^{(1-z)\alpha + z(\alpha / 2)} \right) \right].
    \end{align*}
    This function fulfills the conditions of the three-line theorem (Theorem~\ref{thm:hadamard-3lines}), which we now apply with the choices
\begin{align*}
    p_0 &= 1\\
    p_1 &= 2\\
    \theta &= 2\left( \frac{\alpha-1}{\alpha} \right) \in (0,1].
\end{align*}
Note that these choices imply that $p_\theta = \alpha$ and that $(1-\theta)\alpha + \theta(\alpha / 2) = 1$. We therefore have that
\begin{align}
    |f(\theta)| &= \left| \mbE_h \tr\left[ Y_h \mathcal{N}^h\left( \sigma_E^{\frac{1-\alpha}{2\alpha}} \rho_{AE} \sigma_{E}^{\frac{1-\alpha}{2\alpha}} \right) \right] \right|\\
           &= \mbE_h \left\| \mathcal{N}^h\left(\sigma_E^{\frac{1-\alpha}{2\alpha}} \rho_{AE} \sigma_{E}^{\frac{1-\alpha}{2\alpha}} \right) \right\|_\alpha\\
           &\leqslant M_0^{\frac{2}{\alpha} - 1} M_1^{2\left( \frac{\alpha-1}{\alpha} \right)},   \label{eqn:bound_with_Ms}
\end{align}
where $M_0$ and $M_1$ are defined as in Theorem~\ref{thm:hadamard-3lines}. We now bound these two quantities, starting with $M_0$:
\begin{align}
    M_0 &= \sup_{t \in \mbR} |f(it)|\\
        &= \sup_{t \in \mbR} \left| \mbE_h \tr\left[ Y_h^{it\left( \frac{\alpha}{2(\alpha-1)} \right)} \mathcal{N}^h\left(  \left( \sigma_E^{\frac{1-\alpha}{2\alpha}} \rho_{AE} \sigma_E^{\frac{1-\alpha}{2\alpha}} \right)^{(1-it)\alpha + it(\alpha / 2)} \right) \right] \right|\\
\intertext{Now, notice that $\left\| Y_h^{it\left( \frac{\alpha}{2(\alpha-1)} \right)} \right\|_{\infty} = 1$ since $Y_h$ is Hermitian, and hence we can continue as:}
    &\leqslant \sup_{t \in \mbR} \mbE_h \left\| \mathcal{N}^h\left(  \left( \sigma_E^{\frac{1-\alpha}{2\alpha}} \rho_{AE} \sigma_E^{\frac{1-\alpha}{2\alpha}} \right)^{(1-it)\alpha + it(\alpha / 2)} \right) \right\|_1.
\end{align}
Note now that
\begin{align*}
    \left\| \left( \sigma_E^{\frac{1-\alpha}{2\alpha}} \rho_{AE} \sigma_E^{\frac{1-\alpha}{2\alpha}} \right)^{(1-it)\alpha + it(\alpha / 2)} \right\|_1  &= \left\| \left( \sigma_E^{\frac{1-\alpha}{2\alpha}} \rho_{AE} \sigma_E^{\frac{1-\alpha}{2\alpha}} \right)^{\alpha} \right\|_1\\
                                                                                                              &= \left\| \sigma_E^{\frac{1-\alpha}{2\alpha}} \rho_{AE} \sigma_E^{\frac{1-\alpha}{2\alpha}} \right\|_\alpha^{\alpha}\\
                                                                                                              &= 1.
\end{align*}
This means that we can replace the supremum over $\mbR$ by a maximization over all normal operators with an $1$-norm of 1. This yields the bound we want:
\[ M_0 \leqslant \max_{\eta \in \Normal(\sfA \otimes \sfE) : \| \eta \|_1 = 1} \mbE_h \left\| \mathcal{N}^h(\eta_{AE}) \right\|_1. \]
Note also that if $\rho_{AE}$ is CQ, then so is $\left( \sigma_E^{\frac{1-\alpha}{2\alpha}} \rho_{AE} \sigma_E^{\frac{1-\alpha}{2\alpha}} \right)^{(1-it)\alpha + it(\alpha / 2)}$ and therefore the maximization can be restricted to CQ operators.

The bound on $M_1$ is proven in a very similar manner:
\begin{align*}
    M_1 &= \sup_{t \in \mbR} |f(1+it)|\\
        &= \sup_{t \in \mbR} \left| \mbE_h \tr\left[ Y_h^{(1+it)\left( \frac{\alpha}{2(\alpha-1)} \right)} \mathcal{N}^h\left(  \left( \sigma_E^{\frac{1-\alpha}{2\alpha}} \rho_{AE} \sigma_E^{\frac{1-\alpha}{2\alpha}} \right)^{-it\alpha + (1+it)(\alpha / 2)} \right) \right]\right|.
\end{align*}
        Now, notice that 
\begin{align*}
    \left\| Y_h^{(1+it)\left( \frac{\alpha}{2(\alpha-1)} \right)} \right\|_2 &= \left\| Y_h^{\frac{\alpha}{2(\alpha-1)}} \right\|_2\\
                                                                 &= \tr\left[ Y_h^{\frac{\alpha}{\alpha-1}} \right]^{\frac{1}{2}}\\
                                                                 &= \| Y_h \|_{\frac{\alpha}{\alpha-1}}^{\frac{\alpha}{2(\alpha-1)}}\\
                                                                 &= 1.
\end{align*}
Hence,
\begin{align*}
    M_1 &\leqslant \sup_{t \in \mbR} \mbE_h \left\| \mathcal{N}^h\left(  \left( \sigma_E^{\frac{1-\alpha}{2\alpha}} \rho_{AE} \sigma_E^{\frac{1-\alpha}{2\alpha}} \right)^{-it\alpha + (1+it)(\alpha / 2)} \right) \right\|_2.
\end{align*}
However,
\begin{align*}
    \left\| \left( \sigma_E^{\frac{1-\alpha}{2\alpha}} \rho_{AE} \sigma_E^{\frac{1-\alpha}{2\alpha}} \right)^{-it\alpha + (1+it)(\alpha / 2)} \right\|_2 &= \left\| \left( \sigma_E^{\frac{1-\alpha}{2\alpha}} \rho_{AE} \sigma_E^{\frac{1-\alpha}{2\alpha}} \right)^{\alpha / 2} \right\|_2\\
                                                                                                              &= \tr\left[ \left( \sigma_E^{\frac{1-\alpha}{2\alpha}} \rho_{AE} \sigma_E^{\frac{1-\alpha}{2\alpha}} \right)^\alpha \right]^{\frac{1}{2}}\\
                                                                                                              &= 1.
\end{align*}
We can therefore replace the supremum over $t$ by a maximization over all normal operators with 2-norm of at most one, to get our final bound on $M_1$:
\[ M_1 \leqslant \max_{\eta \in \Normal(\sfA \otimes \sfE) : \| \eta \|_2 = 1} \mbE_h \left\| \mathcal{N}^h(\eta_{AE}) \right\|_2. \]
Once again, if $\rho_{AE}$ is CQ, we can restrict the maximization to CQ operators. Substituting the two bounds we derived on $M_0$ and $M_1$ into \eqref{eqn:bound_with_Ms} finishes the proof.
\end{proof}

We are now ready to state and prove the main theorem:
\begin{theorem}[Main theorem]\label{thm:main}
    Let $\{ \mathcal{R}^h_{A  \rightarrow  C} : h \in \mathcal{H} \}$ be a $\lambda$-randomizing family of channels, and let $\rho_{AE} \in \D(\sfA \otimes \sfE)$ and $\sigma_E \in \D(\sfE)$ such that $\rho_{AE}$ is supported on $\sfA \otimes \Supp(\sigma_E)$. Then,
    \begin{align*}
        \mbE_h \left\| (\mathcal{R}^h - \mathcal{U})(\rho_{AE}) \right\|_1 \leqslant 2^{\frac{2}{\alpha}-1} \cdot 2^{\frac{\alpha-1}{\alpha}(\log|C| - H_\alpha(A|E)_{\rho|\sigma} + 2\log \lambda)}.
    \end{align*}
    Furthermore, if $\rho_{AE}$ is CQ and $\{ \mathcal{R}^h : h \in \mathcal{H}\}$ is a $\lambda$-randomizing family of hash functions, the same bound holds.
\end{theorem}
\begin{proof}
    For any $h \in \mathcal{H}$, let us define the superoperator $\mathcal{N}^h := \mathcal{R}^h - \mathcal{U}$. We can now begin as follows:
    \begin{multline*}
        \mbE_h \left\| (\mathcal{R}^h - \mathcal{U})(\rho_{AE}) \right\|_1\\
        \begin{aligned}
        &= \mbE_h \left\| \mathcal{N}^h(\rho_{AE}) \right\|_1\\
                                                                             &\leqslant \left\|\sigma_E^{\frac{1-\alpha}{2\alpha}}\rho_{AE} \sigma_E^{\frac{1-\alpha}{2\alpha}} \right\|_\alpha \cdot \max_{\omega_{AE} : \left\| \sigma_E^{\frac{1-\alpha}{2\alpha}} \omega_{AE} \sigma_E^{\frac{1-\alpha}{2\alpha}} \right\|_\alpha \leqslant 1} \mbE_h \left\| \mathcal{N}^h(\omega_{AE}) \right\|_1\\
                                                                             &= 2^{\frac{1-\alpha}{\alpha} H_\alpha(A|E)_{\rho|\sigma}} \max_{\omega} \mbE_h \left\| \mathcal{N}^h(\omega_{AE}) \right\|_1\\
                                                                             &= 2^{\frac{1-\alpha}{\alpha} H_\alpha(A|E)_{\rho|\sigma}} \max_{\omega} \mbE_h \left\| \sigma_E^{\frac{\alpha-1}{2\alpha}} \mathcal{N}^h\left( \sigma_E^{\frac{1-\alpha}{2\alpha}} \omega_{AE} \sigma_E^{\frac{1-\alpha}{2\alpha}} \right) \sigma_E^{\frac{\alpha-1}{2\alpha}} \right\|_1\\
                                                                             &\stackrel{(a)}{\leqslant} 2^{\frac{1-\alpha}{\alpha} H_\alpha(A|E)_{\rho|\sigma}} \left\| \ident_C \otimes \sigma_{E}^{\frac{\alpha-1}{\alpha}} \right\|_{\frac{\alpha}{\alpha-1}} \max_{\omega} \mbE_h \left\| \mathcal{N}^h\left( \sigma^{\frac{1-\alpha}{2\alpha}}\omega_{AE} \sigma^{\frac{1-\alpha}{2\alpha}}\right) \right\|_\alpha\\
                                                                             &= 2^{\frac{\alpha-1}{\alpha} (\log|C| - H_\alpha(A|E)_{\rho|\sigma})} \max_{\omega} \mbE_h \left\| \mathcal{N}^h\left( \sigma^{\frac{1-\alpha}{2\alpha}}\omega_{AE} \sigma^{\frac{1-\alpha}{2\alpha}}\right) \right\|_\alpha\\
                                                                             &\stackrel{\text{Lemma \ref{lem:3lines}}}{\leqslant} 2^{\frac{\alpha-1}{\alpha} (\log|C| - H_\alpha(A|E)_{\rho|\sigma})} \left( \max_{\eta_{AE} : \| \eta \|_1 \leqslant 1} \mbE_h \left\| \mathcal{N}^h(\eta_{AE}) \right\|_1 \right)^{\frac{2}{\alpha}-1} \left( \max_{\eta_{AE} : \|\eta\|_2 \leqslant 1} \mbE_h \left\| \mathcal{N}^h(\eta_{AE}) \right\|_2 \right)^{2\left( \frac{\alpha-1}{\alpha} \right)}\\
                                                                             &\leqslant 2^{\frac{2}{\alpha}-1} \cdot 2^{\frac{\alpha-1}{\alpha} (\log|C| - H_\alpha(A|E)_{\rho|\sigma})}  \left( \max_{\eta_{AE} : \|\eta\|_2 \leqslant 1} \mbE_h \left\| \mathcal{N}^h(\eta_{AE}) \right\|_2 \right)^{2\left( \frac{\alpha-1}{\alpha} \right)},
        \end{aligned}
    \end{multline*}
    where $(a)$ follows from Hölder's inequality (Equation~\eqref{eqn:hoelder}), and the last inequality follows from a trivial bound of 2 on the 1-norm term that we can obtain using the triangle inequality and Lemma~\ref{lem:normal-ops-cptp}. To finish the proof, fix $\eta_{AE}$ to any normal operator such that $\| \eta \|_2 \leqslant 1$ (and choose a CQ $\eta_{AE}$ if we are in the CQ case). Then,
    \begin{align}
        \mbE_h \| \mathcal{N}^h(\eta_{AE}) \|_2 &= \mbE_h \| (\mathcal{R}^h - \mathcal{U})(\eta_{AE}) \|_2\\
                                             &\leqslant \lambda \| \eta_{AE} \|_2\\
                                             &\leqslant \lambda,
    \end{align}
    where the first inequality follows from Definition~\ref{def:rand-family}.
    This concludes the proof.
\end{proof}

\section{Sample application: entropy accumulation}\label{sec:entropy-accumulation}
In this section, we show how our main theorem can be applied to the entropy accumulation theorem (EAT)~\cite{dfr16}, which can (among other things) be used to provide security proofs for device-independent QKD protocols with nearly optimal rates. This theorem is used to lower-bound the smooth min-entropy of $n$ random variables that are produced by a sequential process. Its proof follows the outline given in the introduction: we first lower-bound the smooth min-entropy by a Rényi entropy of order $\alpha$, which is then lower-bounded by the von Neumann entropy. Choosing $\alpha \approx 1 + \frac{1}{\sqrt{n}}$ then yields the theorem.

In applications where the bound on the smooth min-entropy is then used for privacy amplification (which is the case for QKD), we can use our main theorem to bypass the smooth min-entropy entirely and directly obtain an error exponent for the trace distance to a perfect key. Furthermore, the resulting proof is overall simpler as it does not involve any smoothing parameter.

To show how this works, we quickly explain a simplified version of the EAT, and then show what happens when we plug in the intermediate statement about Rényi entropies into our main theorem. The EAT provides bounds for states that are generated by step-by-step processes of the form:

    
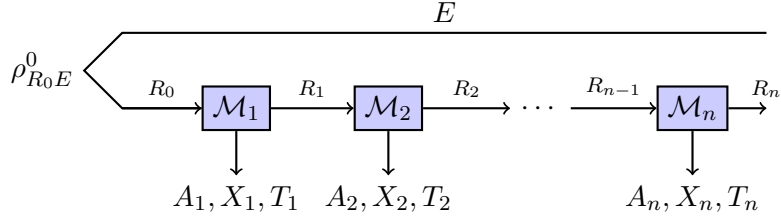
\begin{center}
    \begin{tikzpicture}[thick]
        \draw
            (0, 0) node[porte] (m1) {$\mathcal{M}_1$}
            ++(2, 0) node[porte] (m2) {$\mathcal{M}_2$}
            ++(2, 0) node (dotdotdot) {$\cdots$}
            ++(2, 0) node[porte] (mn) {$\mathcal{M}_n$}
            (m1) ++(0, -1.2) node (a1) {$A_1, X_1, T_1$}
            (m2) ++(0, -1.2) node (a2) {$A_2, X_2, T_2$}
            (mn) ++(0, -1.2) node (an) {$A_n, X_n, T_n$}
            ;
       \draw 
            (m1) ++(-1.5, 0) edge[->] node[midway, above, etiquette] {$R_0$} (m1)
            (m1) edge[->] node[midway, above, etiquette] {$R_1$} (m2)
            (m2) edge[->] node[midway, above, etiquette] {$R_2$} (dotdotdot)
            (dotdotdot) edge[->] node[midway, above, etiquette] {$R_{n-1}$} (mn)
            (m1) edge[->] (a1)
            (m2) edge[->] (a2)
            (mn) edge[->] (an)
            ;
        \draw
            (m1) to ++(-1.5, 0) to  ++(-.5, .5) node[left] {$\rho^0_{R_0 E}$} to ++(.5, .5) coordinate (topright) to node[midway, above] {$E$} ([xshift=.5cm] topright -| mn.east) coordinate (rightwall)
            ;
        \draw[->]
            (mn) to (mn.center -| rightwall) node[above, etiquette] {$R_n$}
            ;
    \end{tikzpicture}
    \end{center}

    In the above, $A_1^n, X_1^n$ are classical systems and that $T_1^n$ are classical bits. Suppose we are interested in lower-bounding $\hmin^\varepsilon(A_1^n|X_1^n E, \wt(T_1^n) = w)$, where $\wt(T_1^n)$ denotes the Hamming weight of the bitstring $T_1^n$ (i.e.~the number of positions that are one).\footnote{In a standard CHSH-based DIQKD protocol, $A_1^n$ would be Alice's output, $X_1^n$ would be Alice's questions, and $T_i$ would indicate whether the CHSH game at position $i$ was won; we then want a bound that depends on how many games were won. The specific details of how this works are not relevant here, but see~\cite{arv16,adfrv17} for a detailed analysis of DIQKD using the EAT.} The EAT provides such a lower bound of the form:
\[ \hmin^\varepsilon(A_1^n|X_1^n E, \wt(T_1^n)=w) \geqslant n f(w) -  \sqrt{n} V \sqrt{1 - 2\log(\varepsilon \Pr[\wt(T_1^n)=w])}, \] 
where $f(\cdot)$ is a tradeoff function that tells us how much entropy we can expect to get given the probability of seeing a one on $T_i$, and $V$ is a constant. We could then use this bound in the usual statement for privacy amplification and get a bound on the trace distance to an ideal key when we apply a randomizing hash function to $A_1^n$.


To see how this picture changes when we use our theorem instead, we reproduce here a simplified version of Proposition 4.5 from~\cite{dfr16}, which is the intermediate statement of the EAT that lower-bounds the Rényi entropy instead of the smooth min-entropy:
\begin{equation}\label{eqn:eat-orig}
     H_\alpha(A_1^n | X_1^n E|\wt(T_1^n) = w) \geqslant nf(w) - n \left( \frac{\alpha-1}{4} \right) V^2 - \frac{\alpha}{\alpha-1}\log \frac{1}{\Pr[\wt(T_1^n)=w]} 
\end{equation}

If we substitute this bound into our main theorem, we get the following: 
\begin{theorem}\label{thm:eat-error-exponent}
    If $\{ \mathcal{R}^h : h \in \mathcal{H}\}$ is a 1-randomizing family of hash functions that produce $nR$ bits at the output, and if $0 < f(w) - R \leqslant \frac{V^2}{2}$, we have that:
    \[ \Pr[\wt(T_1^n) = w] \cdot \mbE_h \left\| \mathcal{R}^h(\rho_{A_1^n X_1^nE|\wt(T_1^n)=w}) - \frac{\ident}{2^{nR}} \otimes \rho_{X_1^nE|\wt(T_1^n)=w} \right\|_1 \leqslant 2 \cdot 2^{-nE(R)} \]
    where the error exponent $E(R)$ is given by:
    \[ E(R) = \frac{1}{2} \left( \frac{f(w)-R}{V} \right)^2. \]
\end{theorem}
This statement is much nicer to work with than what we get by going through the smooth min-entropy. Note in particular that the probability of the event $\wt(T_1^n)=w$ appears in a much more natural way, and we do not have to assume a particular bound for it, as we normally do with the EAT.
\begin{proof}
    We can immediately substitute the bound from \eqref{eqn:eat-orig} into Theorem~\ref{thm:main} to get:
    \begin{multline}
        \mbE_h \left\| \mathcal{R}^h(\rho_{A_1^n X_1^nE|\wt(T_1^n)=w}) - \frac{\ident}{2^{nR}} \otimes \rho_{X_1^nE|\wt(T_1^n)=w} \right\|_1\\
        \begin{aligned}
        &\leqslant 2^{\frac{2}{\alpha}-1} \cdot 2^{\frac{\alpha-1}{\alpha}\left(nR - nf(w) + n\left( \frac{\alpha-1}{4} \right)V^2 + \frac{\alpha}{\alpha-1}\log \frac{1}{\Pr[\wt(T_1^n)=w]} \right)}\\
        &\leqslant \frac{2}{\Pr[\wt(T_1^n)=w]} \cdot 2^{n\frac{\alpha-1}{\alpha}\left(R - f(w) + \left( \frac{\alpha-1}{\alpha} \right)\frac{V^2}{2} \right)}.\label{eqn:error-exp_1}
        \end{aligned}
    \end{multline}
    All that remains to do is to optimize the exponent over $\alpha \in (1,2]$. To do this, let $\beta := \frac{\alpha-1}{\alpha}$, and note that the exponent becomes:
    \begin{align*}
        \beta\left(R - f(w) + \beta\frac{V^2}{2} \right) &= \beta^2 \frac{V^2}{2} - \beta(f(w)-R).
    \end{align*}
    It is easy to see that this expression is minimized whenever $\beta = \frac{f(w)-R}{V^2}$. Substituting this into \eqref{eqn:error-exp_1} yields the bound advertised in the theorem statement. Furthermore, $\alpha \in (1,2]$ if $0 < \beta \leqslant \frac{1}{2}$, which means that this bound applies if $0 < f(w) - R \leqslant \frac{V^2}{2}$ (i.e. only for rates close enough to the first-order term $f(w)$).
\end{proof}

\section{Discussion}\label{sec:discussion}
This result opens up a number of interesting questions for future work. The first is whether we can base all of one-shot information theory on Rényi information measures rather than smooth min/max quantities as has been done so far. This would mean doing the same as was done here for various channel coding problems, multiuser problems, state redistribution, etc. It is already possible to get part of the way there by using the results in this paper and using the ``decoupling principle'', which states that destroying correlations with one system implies the presence of correlations with the purification. If one sets things up correctly, it is possible to use a randomization procedure as defined here to destroy correlations with the channel environment and thus assert the existence of a decoder at the channel output. See~\cite{fred-these} for a comprehensive treatment of this approach to quantum information theory, and Sharma~\cite{s15} gives a Rényi version of this treatment along the lines of what we propose here, but with the caveat that the main decoupling theorem in~\cite{s15} has an additional term that makes its use outside of the iid case difficult. This is unlikely to yield optimal results in terms of error exponents, however, as one already loses a factor of 2 by using Uhlmann's theorem, as one must using this approach. It is also likely to mostly produce results that are expressed as a linear combination of Rényi entropies, instead of having the right quantity (for example, one might get $H_\alpha(A) - H_\beta(A|B)$ instead of $I_\alpha(A;B)$), though this might be mitigated by using chain rules for Rényi entropies such as those found in~\cite{d14}.

One case where this approach is likely to pay off the most is in multiuser Shannon theory problems. In these problems, the smoothing approach leads to the simultaneous smoothing problem, whereby several quantities must be smoothed at once, and we need a single state that is near-optimal for all of them. While there are ways to do this (see~\cite{s18} for example), it considerably complicates the argument, and avoiding the issue altogether using Rényi entropies is a very attractive alternative.

Another question is whether the results obtained here are optimal. Can one get a converse theorem, and if so, how tight is it? It would also be interesting to see if this result sheds any light on the question of how to define the quantum Rényi conditional mutual information (CMI). The main challenge that comes up when trying to generalize the classical definitions is that several non-commuting operators must be multiplied together, and one must choose in which order to do this. For the Rényi CMI, several possibilities have been considered (see~\cite{bsw14} for a systematic study of this problem), but so far none has proven to be fully satisfactory.

It would also be interesting to see to what extent our results can produce better results numerically in realistic DIQKD protocols. In particular, it could be naturally combined with a recent result~\cite{bff20} which defines a new family of Rényi entropies that can be computed by semidefinite programs, and which can be used to compute numerical bounds on tradeoff functions for the EAT. This would provide an ``all-Rényi'' approach to DIQKD security proofs that would likely yield very tight bounds. Note that another route to an all-Rényi approach would be to combine this result with the proof from~\cite{kzf18,kzf18-published}, which gives bounds on the smooth min-entropy generated by device-independent randomness expansion in terms of Rényi entropies---it would thus also be natural there to cut out the min-entropy and work only in terms of Rényi entropies.

Finally, one might try to extend the main theorem to $\alpha > 2$. However, it might be tempting to conjecture that this requires stronger conditions on the randomizing family of channels: the fact that we have an achievability result at $\alpha=2$ is tied to the fact that the channels can randomize the second moment of the state. Going higher might therefore require a family that randomizes higher moments.


\section*{Acknowledgments}\label{sec:acknowledgments}
The author would like to thank Joseph M.~Renes and Andreas Winter as well as the anonymous referees for their helpful comments and suggestions.

\appendix

\section{Technical lemmas}\label{sec:technical-lemmas}
\begin{lemma}\label{lem:normal-ops-cptp}
    Let $\mathcal{N}_{A  \rightarrow  B}$ be a trace non-increasing, completely positive map, and let $M \in \Normal(\sfA)$ be a normal operator. Then,
    \[ \| \mathcal{N}(M) \|_1 \leqslant \| M \|_1. \]
\end{lemma}
\begin{proof}
    First note that $M$ can be decomposed as $M = UD$ for a unitary $U$ and positive semidefinite $D$ such that $U$ and $D$ commute. To see this, write the spectral decomposition of $M$ as
    \begin{align*}
        M = \sum_i \lambda_i e^{i\varphi_i} \proj{\psi_i}
    \end{align*}
    where the $\lambda_i$ and $\varphi_i$ are positive reals. Then,
    \begin{align*}
        D &:= \sum_i \lambda_i \proj{\psi_i}\\
        U &:= \sum_i e^{i\varphi_i} \proj{\psi_i}.
    \end{align*}
    Now, let $\Phi_{AA'} := \sum_{ij} \ket{i i} \bra{jj}$. We then have that
    \begin{align*}
        \| \mathcal{N}(M) \|_1 &= \| \tr_{A'}[\mathcal{N}(U D^{\frac{1}{2}} \Phi D^{\frac{1}{2}})] \|_1\\
                               &= \sup_{V_A : \| V_A \|_\infty \leqslant 1} \left| \tr\left[ V_A \mathcal{N}(U D^{\frac{1}{2}} \Phi D^{\frac{1}{2}}) \right] \right|\\
                               &= \sup_{V_A} \left| \tr\left[ (V_A \otimes U^*_{A'}) \mathcal{N}(D^{\frac{1}{2}} \Phi D^{\frac{1}{2}}) \right] \right|\\
                               &\leqslant \sup_{V_{AA'} : \| V \|_\infty \leqslant 1} \left| \tr\left[ V \mathcal{N}(D^{\frac{1}{2}} \Phi D^{\frac{1}{2}}) \right] \right|\\
                               &= \| \mathcal{N}(D^{\frac{1}{2}} \Phi D^{\frac{1}{2}}) \|_1\\
                               &= \tr[ \mathcal{N}(D^{\frac{1}{2}} \Phi D^{\frac{1}{2}}) ]\\
                               &= \tr[ \mathcal{N}(D)]\\
                               &\leqslant \tr[D]\\
                               &\leqslant \| M \|_1.
    \end{align*}
\end{proof}

\printbibliography

\end{document}